\documentclass [journal,onecolumn,11pt]{IEEEtran}
\usepackage{amsfonts,amsmath,amssymb}
\usepackage{indentfirst, setspace}
\usepackage{url,float}
\usepackage{color}
\usepackage[a4paper,ignoreall]{geometry}
\usepackage{longtable}
\usepackage{graphicx}
\usepackage[a4paper,ignoreall]{geometry}
\usepackage{multicol}
\usepackage{stfloats}
\usepackage{enumerate}
\usepackage{cite}
\usepackage{amsthm}
\usepackage{booktabs}
\usepackage{mathrsfs}
\usepackage{multirow}
\usepackage{amssymb}
\usepackage{verbatim}
\usepackage{cases}
\usepackage[square, comma, sort&compress, numbers]{natbib}
\usepackage[overload]{empheq}
\def\qu#1 {\fbox {\footnote {\ }}\ \footnotetext { From Qu: {\color{red}#1}}}
\def\hqu#1 {}
\def\kq#1 {\fbox {\footnote {\ }}\ \footnotetext { From KangQuan: {\color{blue}#1}}}
\def\hkq#1 {}

\newcommand{\mkq}[1]{{{\color{black}#1}}}

\newtheorem{Th}{Theorem}[section]
\newtheorem{Cor}[Th]{Corollary}
\newtheorem{Prop}[Th]{Proposition}

\newtheorem{Lemma}[Th]{Lemma}
\newtheorem{Def}[Th]{Definition}
\newtheorem{example}{Example}

\newtheorem{Rem}[Th]{Remark}

\newcommand{\tr}{{\rm Tr}}

\newcommand{\gf}{{\mathbb F}}

\makeatletter
\newcommand{\figcaption}{\def\@captype{figure}\caption}
\newcommand{\tabcaption}{\def\@captype{table}\caption}
\makeatother

\begin{document}
	\title{New Results about the Boomerang Uniformity of Permutation Polynomials}
	\author{{Kangquan Li, Longjiang Qu, Bing Sun and Chao Li}
		\thanks{Kangquan Li, Longjiang Qu, Bing Sun and Chao Li are with the College of Liberal Arts and Sciences,
			National University of Defense Technology, Changsha, 410073, China.
			Longjiang Qu is also with the State Key Laboratory of Cryptology, Beijing, 100878, China.
			E-mail: likangquan11@nudt.edu.cn, ljqu\_happy@hotmail.com, happy\_come@163.com, lichao\_nudt@sina.com.
		 This work is supported by the Nature Science Foundation of China (NSFC) under Grant 61722213, 11531002, 61572026,  61772545,  National Key R$\&$D Program of China (No.2017YFB0802000),  and the Open Foundation of State Key Laboratory of Cryptology. }}

	\maketitle{}
	
\begin{abstract}
 In EUROCRYPT 2018,  Cid et al. \cite{BCT2018} introduced a new concept \mkq{on the cryptographic property of} S-boxes: Boomerang Connectivity Table (BCT for short) for evaluating the subtleties of boomerang-style attacks. \mkq{ Very recently, BCT and the boomerang uniformity, the maximum value in BCT, were further studied by Boura and Canteaut \cite{BC2018}. Aiming at providing new insights, we show some new results about  BCT and the boomerang uniformity of permutations in terms of theory and experiment in this paper.  Firstly, we present an equivalent \mkq{technique to compute BCT} and the boomerang uniformity, which seems to be much simpler than the original definition \mkq{ from \cite{BCT2018}.} Secondly, thanks to Carlet's idea \cite{Carlet2018}, we give a characterization of functions $f$ \mkq{from  $\gf_{2}^n$ to itself} with boomerang uniformity $\delta_{f}$ by means of the Walsh transform.  Thirdly, by our method, we consider boomerang uniformities of some specific permutations, mainly the ones with low differential uniformity. Finally, we obtain another class of $4$-uniform BCT permutation polynomials over $\gf_{2^n}$, which is the first binomial.}
\end{abstract}	


\section{Introduction}

Let $p$ be a prime, $n$ any positive integer. We denote by $\gf_{p^n}$ the finite field with $p^n$ elements and by $\gf_p^n$ the $n$-dimensional vector space over $\gf_p$. For any set $E$, we denote the nonzero elements of $E$ by $E\backslash \{0\}$ or $E^{*}$. In this paper, we always identify the vector space $\gf_p^n$ with $\gf_{p^n}$ and consider functions from $\gf_p^n$ to itself as polynomials over $\gf_{p^n}$ for convenience.  A polynomial $f(x)\in\gf_{p^n}[x]$ is called a \emph{permutation polynomial } if the induced mapping $x \to f(x)$ is a permutation \mkq{over} $\gf_{p^n}$.  S-box (over $\gf_2^n$ or $\gf_{2^n}$) is an important component for block ciphers and it is often crucial to require S-boxes to be \mkq{permutations}.  For \mkq{the resistance against} known attacks, several criteria should be satisfied. For example, the Difference Distribution Table (DDT  for short) and the differential uniformity of \mkq{an} S-box \mkq{characterise the resistance of the cryptographic component against differential cryptanalysis} \cite{Biham1991}. \mkq{Furthermore, the differential uniformity, along with many other cryptographic properties of the S-boxes has been extensively studied these years. And a number of results with both theoretical and practical significance have been obtained. It is well known that for any $f$ over $\gf_{2^n}$, the elements of DDT are all even and the minimum of differential uniformities of $f$ is $2$. The functions with differential uniformity $2$ are called  Almost Perfect Nonlinear (APN for short) functions. }


\mkq{The Boomerang attack was proposed by Wagner \cite{Wagner1999} in 1999 and variants of the boomerang attack were later presented \cite{KKS2001,BK2009}. In order to evaluate the subtleties of boomerang-style attacks, in EUROCRYPT 2018, Cid et al. \cite{BCT2018} introduced a new cryptanalysis tool: Boomerang Connectivity Table  and Boomerang Uniformity (see Definition \ref{BCT-def-Fq}).  In \cite{BCT2018}, the authors analyzed the properties of BCT  theoretically, especially the relationship between BCT and DDT. They proved that S-boxes having $2$-uniform DDT always have $2$-uniform BCT and for any choice of $(a,b)$, the value in the BCT is greater than or equal to the one in the DDT. Therefore, for S-boxes,  $2$-uniform BCT permutations are equivalent to APN permutations. Very recently, BCT and the boomerang uniformity were further studied by Boura and Canteaut \cite{BC2018}. Through showing that boomerang uniformity is only an affine equivalent invariant and the classification \cite{LP2007} about all differentially $4$-uniform permutations of $4$ bits, Boura and Canteaut completely characterized the BCT of such permutations. In addition, they also obtained  the boomerang uniformities  of the inverse function and the Gold function over $\gf_{2^n}$.
}

\mkq{To better reveal the guidance of the newly proposed cryptographic criteria on how to design S-boxes, 
	in this paper, we further explore novel properties about BCT and the boomerang uniformity of permutations over $\gf_{2^n}$ theoretically and experimentally. Firstly, we give a new method about computing BCT and the boomerang uniformity of permutations, which is simpler than the original one. In detail, our definition transforms the problem solving an extremely complicated equation with a permutation and its inverse into that of solving an equation system including two simpler equations with only the permutation in order to compute the BCT and the boomerang uniformity of the permutation. After this transformation, not only can we compute BCT and the boomerang uniformity of permutations more simply, but we can study their properties, such as the characterization of the boomerang uniformity by the Walsh transform and so on, more easily. Moreover, using our new method, we compute boomerang uniformities of some specific permutations, mainly the ones with low differential uniformity and obtain another class of $4$-uniform BCT permutation polynomials over $\gf_{2^n}$,  
	 which is the first binomial up to now. } 
 
The rest of this paper is organized as follows. \mkq{In Section 2, we first recall the original definition about BCT and the boomerang uniformity of permutations from \cite{BCT2018}. Moreover, in consideration of the complexity of computing BCT and the boomerang uniformity, we give an equivalent formula to compute them, which seems simpler and can be generalized (not restricted to permutations).  Section 3  gives a characterization of $\delta_{f}$-uniform BCT functions by the Walsh transform. Next, we compute boomerang uniformities of some permutation polynomials with low differential uniformity over $\gf_{2^n}$ theoretically and experimentally and obtain another class of $4$-uniform BCT permutations in Sections 4 and 5, respectively. Finally, Section 6 is a conclusion.
}

\section{\mkq{New method of computing BCT and Boomerang Uniformity}}

In \cite{BCT2018}, Cid et al. introduced the concept of Boomerang Connectivity Table  of a permutation $f$ from $\gf_{2}^n$ to itself \mkq{ as follows, which is also suitable for the case $\gf_{2^n}$ clearly.

\begin{Def}
	\cite{BCT2018}
	\label{BCT-def-Fq}
	Let $f$ from $\gf_{2}^n$ to $\gf_{2}^n$ be an invertible function, and $a,b\in\gf_{2}^n$. The Boomerang Connectivity Table (BCT) of $f$ is given by a $2^n\times 2^n$ table $T$, in which the entry for the $(a,b)$ position is given by 
	\begin{equation}
	\label{T(a,b)}
	T(a,b)=\#\{x\in\gf_2^n | f^{-1}(f(x)+a)+f^{-1}(f(x+b)+a) = b \}.
	\end{equation}
 Moreover, for any $a,b\in\gf_2^n\backslash\{0\}$, the value
	$$\delta_f = \max \limits_{a,b\in\gf_2^n\backslash\{0\}} \#\{x\in\gf_2^n | f^{-1}(f(x)+a)+f^{-1}(f(x+b)+a) = b \},  $$
  is called {the boomerang uniformity}  of $f$, or we call $f$ is a {$\delta_f$-uniform BCT}  function.
\end{Def}
  We note that Definition \ref{BCT-def-Fq} is only suitable for invertible functions, i.e., permutations.}
According to the definition of BCT and the boomerang uniformity, for a permutation $f(x)$ over $\gf_{2^n}$, it is necessary to obtain the compositional inverse $f^{-1}(x)$ of $f(x)$ over $\gf_{2^n}$ if we want to compute  BCT and the boomerang uniformity of $f(x)$. However, given a permutation polynomial $f(x)$ over $\gf_{p^n}$, it is a hard problem to compute the compositional inverse with explicit form of $f(x)$ over $\gf_{p^n}$. Besides \mkq{many} classical classes such as monomials, linearized polynomials, and Dickson polynomials, there are few classes of permutation polynomials whose compositional inverses have been obtained in explicit forms. Some results about compositional inverses can be referred to \cite{TW2014,LQW2018,ZYP2015,WL2013}. In addition, it is general that the compositional inverse of a permutation polynomial $f$ with a simple form has a complex form, increasing the difficulty of computing  BCT and the boomerang uniformity  of $f$. Therefore, it seems interesting and  meaningful to compute  BCT and the boomerang uniformity of $f(x)$ without $f^{-1}(x)$. 

\mkq{In the following, we present an equivalent formula to compute BCT and the boomerang uniformity without knowing $f^{-1}(x)$ and $f(x)$ simultaneously.
 
Let $a,b \in \gf_{2^n}^{*}$ and  $f(x)$ be a permutation polynomial over $\gf_{2^n}$.  Let $y=x+b$ in (\ref{T(a,b)}). 
Then the following equation system
\begin{equation}
\label{BCT-equ-sys}
\left\{
\begin{array}{lr}
	f^{-1}(f(x)+a)+ f^{-1}(f(y)+a) = b, &  \\
x + y = b. &
\end{array}
\right.
\end{equation}
has $T(a,b)$  solutions in $\gf_{2^n}\times \gf_{2^n}$.
Furthermore, let $X=f(x)$ and $Y=f(y)$. Then we have 
\begin{equation}
\label{BCT-equ-sys-1}
\left\{
\begin{array}{lr}
	f^{-1}(X+a) + f^{-1}(Y+a) = b ,  \\
  f^{-1}(X) + f^{-1}(Y) = b    
\end{array}
\right.
\end{equation}
from  the equation system (\ref{BCT-equ-sys}) and the numbers of solutions of equation systems (\ref{BCT-equ-sys}) and (\ref{BCT-equ-sys-1}) are the same since $f$ is a permutation polynomial over $\gf_{2^n}$. Hence, it is sufficient to compute the solutions of the equation system (\ref{BCT-equ-sys-1}) if we want to obtain  BCT and the boomerang uniformity of a given permutation polynomial $f(x)\in\gf_{2^n}[x]$. } 

\mkq{\begin{Th}
	\label{BCT-def-2}
	Let $f(x)\in\gf_{2^n}[x]$ be a permutation polynomial over $\gf_{2^n}$, $f^{-1}(x)$ be the compositional inverse of $f(x)$ over $\gf_{2^n}$ and $a,b\in\gf_{2^n}$. Then the BCT of $f(x)$ can be given by a ${2^n}\times {2^n}$ table $T$, in which the entry $T(a,b)$ for the $(a,b)$ position is given by the number of solutions of the equation system (\ref{BCT-equ-sys-1}). 
\end{Th}}

Let $f(x)\in\gf_{2^n}[x]$ be a permutation polynomial over $\gf_{2^n}$, $f^{-1}(x)$ be the compositional inverse of $f(x)$ over $\gf_{2^n}$ and $a,b\in\gf_{2^n}$. Assume that $T$ and $T^{'}$ are the BCTs of $f(x)$ and $f^{-1}(x)$, respectively.  From \cite[Proposition 2]{BC2018}, we know that for any $a,b\in\gf_{2^n}^{*}$, $T(a,b)=T^{'}(b,a)$. Together with Theorem \ref{BCT-def-2}, we have
\mkq{\begin{Th}
		\label{boomerang-def}
		Let  $f(x)$ be a permutation polynomial over  $\gf_{2^n}$. Then the boomerang uniformity of $f(x)$, given by $\delta_f$, is the maximum of numbers of solutions in $\gf_{2^n}\times\gf_{2^n}$ of the following equation system		
   \begin{equation}
\label{boomerang-uniformity}
\left\{
\begin{array}{lr}
f(x+a) + f(y+a) = b,  \\
f(x) + f(y) = b 
\end{array}
\right.
\end{equation}
      for any $a,b\in\gf_{2^n}^{*}$. 
\end{Th}
}

\begin{Rem}

\mkq{\emph{  About out new method to compute BCT and the boomerang uniformity, the first advantage is that we do not have to figure out the compositional inverse of $f(x)$. In addition, in the original definition about BCT and the boomerang uniformity of $f(x)$ over $\gf_{2}^n$  by Cid et al. \cite{BCT2018}, i.e., Definition \ref{BCT-def-Fq}, it is assumed that  $f(x)$ is a permutation. Nevertheless, it is clear from Theorem \ref{boomerang-def} that the condition with permutation property is not necessary. Finally, the boomerang uniformity can be generalized to any vector space, i.e.,  $\gf_{p}^n$ with $p$ odd, where we should notice  that the equation system (\ref{boomerang-uniformity}) becomes }
\begin{equation*}
\left\{
\begin{array}{lr}
	f(y+a) - f(x+a) = b,  \\
	f(y) - f(x) = b.   
\end{array}
\right.
\end{equation*}
\emph{ However, in this paper, we mainly consider the boomerang uniformity of permutations over $\gf_{2^n}$.} }

\end{Rem}

\begin{Rem}
	\label{remark}
\emph{If $(x_0,y_0)$ is a solution in $\gf_{2^n}\times\gf_{2^n}$ of the equation system (\ref{boomerang-uniformity}), then $y_0\neq x_0$ since $b\neq 0$. Hence $(y_0,x_0)$ must be a distinct solution of  (\ref{boomerang-uniformity}). Thus for any $f$, the elements in BCT of $f$ must be even. Moreover, $(x_0+a,y_0+a)$ and $(y_0+a,x_0+a)$ are another two solutions of  (\ref{boomerang-uniformity}) if $x_0+a\neq y_0$. That is to say,  $(x_0,y_0), (y_0,x_0), (x_0+a,y_0+a)$ and $(y_0+a,x_0+a)$ are four different solutions of  (\ref{boomerang-uniformity}) if  $x_0+a\neq y_0$ and the boomerang uniformity of $f$, we think, is probably more than four. In other words, constructing $4$-uniform BCT permutations seems not so easy. } 

\mkq{\emph{From Theorem \ref{boomerang-def}, it is easy to see that BCT is an affine equivalent invariant, but is nor an EA and CCZ equivalence invariant, as already shown in \cite{BC2018}. Recall that two functions $f$ and $f^{'}$ from $\gf_2^n$ to itself are called EA equivalent if $f^{'}=A_1\circ f\circ A_2+A$, where $A$ is affine and $A_1,A_2$ are affine permutations. In particular, when $A=0$, $f$ and $f^{'}$ are called affine equivalent. Assume that $A$ is affine and $A_1,A_2$ are affine permutations. If $f'(x)+f'(y)=b$, then we have $A_1\circ f\circ A_2(x)+A(x) + A_1\circ f\circ A_2(y)+A(y) =b,$ and $f\circ A_2(x)  + f\circ A_2(y)  + L_1^{-1}\circ L(x+y) =L_1^{-1}(b)$, which is equivalent to  $f(X) + f(Y) = L_1^{-1}(b)$ if $L = 0$, where $L_1$ and $L$ are the linear part of $A_1$ and $A$, respectively, $X=A_2(x)$ and $Y=A_2(y)$. Therefore, BCT is an affine equivalent invariant. However, it is clear that one can not build a similar equation if $L\neq 0$. Hence BCT is nor an EA or CCZ equivalence invariant. }}

\end{Rem}

%
%

At the end of this section, consider the boomerang uniformities of permutation polynomials with special forms, such as monomials and quadratic permutations.

\begin{Def}
	\cite{Bud2014}
    A function $f$  from $\gf_{p^n}$ to itself is 
	\begin{enumerate}
		\item linear if $$f(x)=\sum_{0\le i < n}a_i x^{p^i}, \qquad  a_i\in\gf_{p^n};$$
		\item affine if $f$ is a sum of a linear function and a constant;
		\item Dembowski-Ostrom polynomial (DO polynomial) if 
		  $$f(x)=\sum_{0\le k< j \leq n-1}a_{kj}x^{p^k+p^j}, ~~~~~~~ a_{kj}\in\gf_{p^n};$$
		\item quadratic if it is a sum of a DO polynomial and an affine function.
	\end{enumerate}
\end{Def} 

\begin{Prop}
	\label{mon-sim}
	Let $f(x)=x^d\in\gf_{2^n}[x]$. Then the boomerang uniformity  of $f(x)$ is $\delta_{f}=\max \limits_{b\in\gf_{2^n}^{*}} T(1,b)$, where $T(1,b)$ is the number of solutions over $\gf_{2^n}\times \gf_{2^n}$ of the following equation system 
	\begin{subequations}  
		\begin{empheq}[left={\empheqlbrace\,}]{align}
		f(x+1) &+ f(y+1) = b    \notag \\ 
f(x)& + f(y) = b.   ~~~~ ~  \notag 
		\end{empheq}         
	\end{subequations}
\end{Prop}  

\begin{proof}
\mkq{	Let $x=aX$ and $y=aY$. Then the result follows directly  from Theorem  \ref{boomerang-def} and $f(x)=x^d.$}
\end{proof}

\begin{Prop}
	\label{Quadratic}
	Let $f(x)$ be a quadratic differentially $\Delta$-uniform permutation polynomial over $\gf_{2^n}$. Then the boomerang uniformity $\delta_{f}$ of $f$ satisfies $\Delta \le \delta_{f}\le \Delta(\Delta-1)$. Particularly, if $\Delta=2$, then $\delta_{f}=2$;	if $\Delta=4$, $4\le\delta_{f}\le 12.$ 
\end{Prop}

\begin{proof}
 we only prove the first inequality and it suffices to prove its right part. From \mkq{Theorem \ref{boomerang-def}}, we know that $\delta_{f}=\max \limits_{a,b\in\gf_{2^n}^{*}} T(a,b)$, where $T(a,b)$ is the number of solutions over $\gf_{2^n}\times \gf_{2^n}$ of the following equation system                                               
	\begin{subequations} 	\label{quad-1}
		\renewcommand\theequation{\theparentequation.\arabic{equation}}     
		\begin{empheq}[left={\empheqlbrace\,}]{align}
		f(x+a) + f(y&+a) = f(x)+f(y)    \label{quad-1:1A} \\ 
f(x)& + f(y) = b.       \label{quad-1:1B}
		\end{empheq}
	\end{subequations}

  Let $D_{a}f(x)= f(x+a)+f(x)$. Then  for any $a\in\gf_{2^n}^{*}$, $D_{a}f(x)$ is linear since $f(x)$ is quadratic.  Moreover,  $D_af(x)$ is $\Delta$-to-$1$ since the differential uniformity of $f(x)$ is $\Delta$. Together with  (\ref{quad-1:1A}), i.e., $D_af(x)=D_af(y)$, we have $y=x+\alpha_i$, where $i=1,\cdots,\Delta-1$ and $\alpha_i\neq\alpha_j$ for any $i\neq j$. For any $1\le i\le \Delta-1$,  (\ref{quad-1:1B}), i.e., $f(x)+f(x+\alpha_i)=b$ has at most $\Delta$ solutions in $\gf_{2^n}$. Therefore, the equation system (\ref{quad-1}) has at most $\Delta(\Delta-1)$ solutions in $\gf_{2^n}\times\gf_{2^n}$. That is to say, $\delta_f\le \Delta(\Delta-1)$. We finish the proof.
\end{proof}

\begin{Rem}
	\emph{\mkq{Proposition \ref{Quadratic} is a generalization of \cite[Proposition 7]{BC2018}, which only gave the result for the case $\Delta=4$. Moreover, even for this case, our proof seems to be much simple.  }}
\end{Rem}

\section{Characterizations of $\delta_f$-uniform BCT functions by the Walsh transform}

It is well known that \mkq{ for any $u\in\gf_{2}^n$ and $v\in\gf_{2}^n\backslash\{0\}$}, the Walsh transform of $f$ \mkq{from $\gf_2^n$ to itself} is 
$$W_f(u,v)=\sum_{x\in\gf_{2}^n}(-1)^{u\cdot x + v\cdot f(x)}.$$
In this section, we consider the characterizations of $\delta_{f}$-uniform BCT functions by the Walsh transform. The main idea is from Carlet \cite{Carlet2018} who characterized the differential uniformity of vectorial functions by means of the Walsh transform. Let $$T(a,b)=\#\{ (x,y)\in\gf_2^n\times \gf_2^n | f(x+a)+f(y+a)=b ~~\text{and}~~ f(x)+f(y)=b  \}.$$
It is clear that $f$ is \mkq{an at most} $\delta_{f}$-uniform BCT function if and only if, for any $a,b\in\gf_{2}^n\backslash\{0\}$, we have $T(a,b)\in\{0,2, \cdots, \delta_{f}\}$. Let $\phi(x)=\sum_{j\ge0}A_jx^j$ be any polynomial over $\mathbb{R}$ such that $\phi(x)=0$ for $x=0,2,\cdots,\delta_{f}$ and $\phi(x)>0$ for every even $x\in\{\delta_{f}+2, \cdots, 2^n  \}$. Hence for any $f$ and $a,b\in\gf_{2}^n\backslash\{0\}$, we have 
$$\sum_{j\ge0}A_j (T(a,b))^j\ge0, $$
and $f$ is \mkq{an at most} $\delta_{f}$-uniform BCT function if and only if this inequality is an equality for any $a,b\in\gf_{2}^n\backslash\{0\}$. Furthermore, for any $f$, we have 
$$\sum_{j\ge0}A_j \sum_{a,b\in\gf_{2}^n\backslash\{0\}} (T(a,b))^j\ge 0,  $$
and $f$ is \mkq{an at most}  $\delta_{f}$-uniform BCT function if and only if this inequality is an equality. We now characterize $\delta_f$-uniform BCT functions by the Walsh transform. 

\begin{Lemma}
	\label{lemma-T(a,b)-walsh}
	For any $j\ge1$, we have 
	\begin{eqnarray*}
		&  & \sum_{a,b\in\gf_{2}^n\backslash\{0\}}\left(T(a,b)\right)^j  \\
		& = & 2^{2n-4nj}\sum_{{\mbox{\tiny$\begin{array}{c}
					\alpha_1,\cdots,\alpha_j,\beta_1,\cdots,\beta_j\in\gf_{2}^n \\
					\gamma_1,\cdots,\gamma_j,\eta_1,\cdots,\eta_j\in\gf_{2}^n\\
					\sum_{i=1}^{j}(\alpha_i+\beta_i)=0, \sum_{i=1}^{j}(\gamma_i+\eta_i)=0\\
					\end{array}$}} } \prod_{i=1}^{j}W_f(\gamma_i,\alpha_i)W_f(\eta_i,\alpha_i)W_f(\gamma_i,\beta_i)W_f(\eta_i,\beta_i)-2^{nj}(2^{n+1}-1).
	\end{eqnarray*}
\end{Lemma}
\begin{proof}
	\mkq{	Firstly, it is clear that $$T(a,b)=2^{-2n}\sum_{\alpha,\beta,x,y\in\gf_{2}^n}(-1)^{\alpha\cdot(f(x+a)+f(y+a)+b)+\beta\cdot(f(x)+f(y)+b)}$$  
	since $\sum_{\alpha\in\gf_{2}^n}(-1)^{\alpha\cdot(f(x+a)+f(y+a)+b)}$ (or $\sum_{\beta\in\gf_{2}^n}(-1)^{\beta\cdot(f(x)+f(y)+b)}$) is nonzero for $f(x+a)+f(y+a)=b$ (resp. $f(x)+f(y)=b$) only and takes  $2^n$.
 Moreover, we have 
		\begin{eqnarray*}
		T(a,b) &=& 2^{-4n} \sum_{{\mbox{\tiny$\begin{array}{c}
					\alpha,\beta,\gamma,\eta\in\gf_{2}^n \\
					x,y,z,w\in\gf_{2}^n\\
					\end{array}$}} } (-1)^{\alpha\cdot(f(z)+f(w)+b)+\beta\cdot(f(x)+f(y)+b)+\gamma\cdot(z+x+a)+\eta\cdot(w+y+a)} \\
			   &=& 2^{-4n}\sum_{{\mbox{\tiny$\begin{array}{c}
			   			\alpha,\beta,\gamma,\eta\in\gf_{2}^n \\
			   			x,y,z,w\in\gf_{2}^n\\
			   			\end{array}$}}} (-1)^{\gamma\cdot z+ \alpha\cdot f(z)} (-1)^{\eta\cdot w+\alpha\cdot f(w)} (-1)^{\gamma\cdot x+\beta\cdot f(x)} (-1)^{\eta\cdot y+\beta\cdot f(y)} (-1)^{(\alpha +\beta)\cdot b+(\gamma+\eta)\cdot a} \\
		   	   &=& 2^{-4n} \sum_{	\alpha,\beta,\gamma,\eta\in\gf_{2}^n} W_f(\gamma,\alpha)W_f(\eta,\alpha)W_f(\gamma,\beta)W_f(\eta,\beta) (-1)^{(\alpha +\beta)\cdot  b+(\gamma+\eta)\cdot a} ,
		\end{eqnarray*}

	Similarly, for an integer $j \ge 1$, 
		$$
	\sum_{a,b\in\gf_{2}^n}\left(T(a,b)\right)^j=2^{-2nj}\sum_{a,b\in\gf_{2}^n}\sum_{{\mbox{\tiny$\begin{array}{c}
				\alpha_1,\cdots,\alpha_j,\beta_1,\cdots,\beta_j\in\gf_{2}^n \\
				x_1,\cdots,x_j,y_1,\cdots,y_j\in\gf_{2}^n
				\end{array}$}} }(-1)^{\sum_{i=1}^{j}\alpha_i\cdot (f(x_i+a)+f(y_i+a)+b)+\beta_i\cdot (f(x_i)+f(y_i)+b)}
	$$
	and 
	\begin{eqnarray*}
		&  & 2^{4nj}\sum_{a,b \in\gf_{2}^n}\left(T(a,b)\right)^j  \\
		& = & \sum_{a,b \in\gf_{2}^n}\sum_{{\mbox{\tiny$\begin{array}{c}
					\alpha_1,\cdots,\alpha_j,\beta_1,\cdots,\beta_j\in\gf_{2}^n \\
					x_1,\cdots,x_j,y_1,\cdots,y_j\in\gf_{2}^n\\
					\gamma_1,\cdots,\gamma_j,\eta_1,\cdots,\eta_j\in\gf_{2}^n\\
					z_1,\cdots,z_j,w_1,\cdots,w_j\in\gf_{2}^n\\
					\end{array}$}} }(-1)^{\sum_{i=1}^{j}\alpha_i\cdot(f(z_i)+f(w_i)+b)+\beta_i\cdot(f(x_i)+f(y_i)+b)+\gamma_i\cdot(z_i+x_i+a)+\eta_i\cdot(w_i+y_i+a)} \\
		& = & \sum_{{\mbox{\tiny$\begin{array}{c}
					\alpha_1,\cdots,\alpha_j,\beta_1,\cdots,\beta_j\in\gf_{2}^n \\
					\gamma_1,\cdots,\gamma_j,\eta_1,\cdots,\eta_j\in\gf_{2}^n\\
					\end{array}$}} } \prod_{i=1}^{j}W_f(\gamma_i,\alpha_i)W_f(\eta_i,\alpha_i)W_f(\gamma_i,\beta_i)W_f(\eta_i,\beta_i) \left(\sum_{a,b \in\gf_{2}^n}(-1)^{\sum_{i=1}^{j}(\alpha_i+\beta_i)\cdot b+(\gamma_i+\eta_i)\cdot a}\right) .
	\end{eqnarray*}
	Because $\sum_{a,b \in\gf_{2}^n}(-1)^{\sum_{i=1}^{j}(\alpha_i+\beta_i)\cdot b+(\gamma_i+\eta_i)\cdot a}$ is nonzero only when $\sum_{i=1}^{j}(\alpha_i+\beta_i)=0 $ and $\sum_{i=1}^{j}(\gamma_i+\eta_i)=0 $ hold at the same time and takes $2^{2n}$, we have 
\begin{eqnarray*}
	&  & 2^{4nj}\sum_{a,b \in\gf_{2}^n}\left(T(a,b)\right)^j  \\
	& = & 2^{2n}\sum_{{\mbox{\tiny$\begin{array}{c}
				\alpha_1,\cdots,\alpha_j,\beta_1,\cdots,\beta_j\in\gf_{2}^n \\
				\gamma_1,\cdots,\gamma_j,\eta_1,\cdots,\eta_j\in\gf_{2}^n\\
				\sum_{i=1}^{j}(\alpha_i+\beta_i)=0, \sum_{i=1}^{j}(\gamma_i+\eta_i)=0\\
				\end{array}$}} } \prod_{i=1}^{j}W_f(\gamma_i,\alpha_i)W_f(\eta_i,\alpha_i)W_f(\gamma_i,\beta_i)W_f(\eta_i,\beta_i)
\end{eqnarray*}
On the other hand, $$\sum_{a,b \in\gf_{2}^n}\left(T(a,b)\right)^j=\sum_{a,b \in\gf_{2}^n\backslash\{0\}}\left(T(a,b)\right)^j+2^{nj}(2^{n+1}-1).$$
Thus 
\begin{eqnarray*}
	&  & \sum_{a,b \in\gf_{2}^n\backslash\{0\}}\left(T(a,b)\right)^j  \\
	& = & 2^{2n-4nj}\sum_{{\mbox{\tiny$\begin{array}{c}
				\alpha_1,\cdots,\alpha_j,\beta_1,\cdots,\beta_j\in\gf_{2}^n \\
				\gamma_1,\cdots,\gamma_j,\eta_1,\cdots,\eta_j\in\gf_{2}^n\\
				\sum_{i=1}^{j}(\alpha_i+\beta_i)=0, \sum_{i=1}^{j}(\gamma_i+\eta_i)=0\\
				\end{array}$}} } \prod_{i=1}^{j}W_f(\gamma_i,\alpha_i)W_f(\eta_i,\alpha_i)W_f(\gamma_i,\beta_i)W_f(\eta_i,\beta_i)-2^{nj}(2^{n+1}-1).
\end{eqnarray*}
	
}
\end{proof}

Note that for $j=0$, we have $\sum_{a,b \in\gf_{2}^n\backslash\{0\}}\left(T(a,b)\right)^j=(2^n-1)^2$. Therefore, we have the following theorem.

\begin{Th}
	Let $n, \delta$ be positive integers, where $\delta$ is even, and let $f$ be any permutation over $\gf_{2^n}$. Let $\phi(x)=\sum_{j\ge0}A_jx^j$ be any polynomial over $\mathbb{R}$ such that $\phi(x)=0$ for $x=0,2,\cdots,\delta$ and $\phi(x)>0$ for every even $x\in\{\delta+2, \cdots, 2^n \}$. Then we have 
	  $$(2^n-1)^2 A_0 + \sum_{j\ge1}A_j \sum_{a,b\in\gf_{2}^n\backslash\{0\}} (T(a,b))^j\ge 0, $$
	  where $\sum_{a,b\in\gf_{2}^n\backslash\{0\}}(T(a,b))^j$ is given in Lemma \ref{lemma-T(a,b)-walsh} for any $j\ge1$.	 Furthermore, this inequality is an equality if and only if the boomerang uniformity of $f$ is $\delta$. 
\end{Th}

\subsection{Characterizations of $2$-uniform BCT functions by the Walsh transform}
Let $\phi(x)= x(x-2)=x^2-2x$, which satisfies that $\phi(x)=0$ for $x=0,2$ and $\phi(x)>0$ for every even $x\in \{4, 6, \cdots, 2^n  \}$. Then
$$\sum_{a,b\in \gf_{2}^n\backslash\{0\}}(T(a,b))^2-2\sum_{a,b\in\gf_{2}^n\backslash\{0\}}T(a,b)\ge0,$$
and $f$ is a  $2$-uniform BCT function if and only if the above inequality is an equality.

From Lemma \ref{lemma-T(a,b)-walsh}, we have 
	\begin{eqnarray*}
	&  & \sum_{a,b \in   \gf_{2}^n\backslash\{0\}}T(a,b) =  2^{-2n}\sum_{{\mbox{\tiny$\begin{array}{c}
				\alpha, \gamma\in \gf_{2}^n\\
				\end{array}$}} }W_f(\alpha,\gamma)^4-2^{n}(2^{n+1}-1).
\end{eqnarray*}
and 
	\begin{eqnarray*}
	&  & \sum_{a,b \in\gf_{2}^n\backslash\{0\}}\left(T(a,b)\right)^2  \\
	& = & 2^{-6n}\sum_{{\mbox{\tiny$\begin{array}{c}
				\alpha_1,\alpha_2,\beta_1,\beta_2\in\gf_{2}^n \\
				\gamma_1,\gamma_2,\eta_1,\eta_2\in\gf_{2}^n\\
				\sum_{i=1}^{2}(\alpha_i+\beta_i)=0, \sum_{i=1}^{2}(\gamma_i+\eta_i)=0\\
				\end{array}$}} } \prod_{i=1}^{2}W_f(\gamma_i,\alpha_i)W_f(\eta_i,\alpha_i)W_f(\gamma_i,\beta_i)W_f(\eta_i,\beta_i)-2^{2n}(2^{n+1}-1) \\	
\end{eqnarray*}

\begin{Th}
	\label{2-uniform}
	Let $f$ be any function from $\gf_{2}^n$ to itself. Then 
		\begin{eqnarray*}
		&  & \sum_{{\mbox{\tiny$\begin{array}{c}
					\alpha_1,\alpha_2,\beta_1,\beta_2\in\gf_{2}^n \\
					\gamma_1,\gamma_2,\eta_1,\eta_2\in\gf_{2}^n\\
					\sum_{i=1}^{2}(\alpha_i+\beta_i)=0, \sum_{i=1}^{2}(\gamma_i+\eta_i)=0\\
					\end{array}$}} } \prod_{i=1}^{2}W_f(\gamma_i,\alpha_i)W_f(\eta_i,\alpha_i)W_f(\gamma_i,\beta_i)W_f(\eta_i,\beta_i)\\
		& \ge & 2^{4n+1}\sum_{{\mbox{\tiny$\begin{array}{c}
					\alpha, \gamma\in \gf_{2}^n\\
					\end{array}$}} }W_f(\gamma,\alpha)^4 + 2^{9n+1} -5\cdot  2^{8n} + 2^{7n+1}.\\	
	\end{eqnarray*}
   Moreover, $f$ is a $2$-uniform BCT function if and only if  the above inequality is an equality.
 \end{Th}

\mkq{Suppose $f$ is a function from $\gf_{2}^n$ to itself. Since the differential uniformity of $f$ is $2$ if and only if $f$ is a $2$-uniform BCT function from \cite{BCT2018}, Theorem \ref{2-uniform} gives another characterization of APN functions by means of the Walsh transform. }

\section{The boomerang uniformity of  some permutation polynomials \newline with low differential uniformity}

\mkq{Before this section, we introduce some notations that will be used in the following. Let $\omega$ be an element of $\gf_{2^2}\backslash\gf_{2}$ and $\tr_{2^n}(\cdot)$ denote the absolute trace function from $\gf_{2^n}$ to $\gf_{2}$. For any $\gamma\in\gf_{2^n}^{*}$, we assume $\mathrm{Ord}(\gamma)$ is the order of $\gamma$, i.e., the minimum positive integer $k$ such that $\gamma^k=1$.}

\subsection{APN permutations}

From \cite{BCT2018}, we know that the differential uniformity of permutation $f$ over $\gf_{2^n}$ is $2$ if and only if $f$ is a $2$-uniform BCT permutation. Thus constructing $2$-uniform BCT permutations over $\gf_{2^n}$  is equivalent to constructing APN permutations over $\gf_{2^n}$. In the case $n$ is even, there is only a sporadic APN permutation over $\gf_{2^6}$ presented by Dillon et al. \cite{Dillon2010} up to now and we call it as Dillon's Permutation. Therefore, when $n$ is even, there is also one $2$-uniform BCT permutation (Dillon's Permutation) over $\gf_{2^n}$. As for the case $n$ is odd, there are many infinite classes of APN permutations as follows and thus also $2$-uniform BCT permutations. In fact,  \cite{BCC2006} proved that if $x^d$ over $\gf_{2^n}$ is an APN, then 
\begin{equation*}
\gcd(d,2^n-1)= \left\{
\begin{aligned}
1, & ~~\text{if}~~ n ~~\text{is odd};  \\
3, & ~~\text{if}~~ n ~~\text{is even}. \\ 
\end{aligned}
\right.
\end{equation*}
The above result shows that APN power functions over $\gf_{2^n}$ must be permutations when $n$ is odd while those can not be permutations when $n$ is even. In Table I, we list  all current APN \mkq{functions}, i.e., $2$-uniform BCT permutation monomials over $\gf_{2^n}$, where $n$ is odd.

   \begin{table}[!htbp]
	\label{2-uniform-mon}
	\caption{$2$-uniform BCT permutation monomials over $\gf_{2^n}$, $n$ odd}
	\centering
	\begin{tabular}{cccc}	
		\toprule
		Function & Expression & Conditions & Ref. \\
		\midrule
		Gold &	$x^{2^i+1}$ & $\gcd(n,i)=1$  &  \cite{Gold1968,Nyberg1994} \\
		Kasami &	$x^{2^{2i}-2^i+1}$ & $\gcd(n,i)=1$ & \cite{Kasami1971}  \\
		Welch & $x^{2^k+3}$ & $n=2k+1$ & \cite{Dobbertin1999} \\
		Niho-1 & $x^{2^k+2^{k/2}-1}$ & $n=2k+1$, $k$ even & \cite{Dobbertin1999-1} \\
		Niho-2 & $x^{2^k+2^{(3k+1)/2}-1}$ & $n=2k+1$, $k$ odd & \cite{Dobbertin1999-1} \\  		
		Inverse &	$x^{-1}$ & $n$ odd &  \cite{Gold1968,BD1993} \\
		Dobbertin &	$x^{2^{4k}+2^{3k}+2^{2k}+2^k-1}$ & $n=5k$ & \cite{Dobbertin2001}  \\
		\bottomrule
	\end{tabular}
\end{table}

 In addition, there are also many APN functions with dominant expressions over $\gf_{2^n}$, such as \cite{BCL2008,BCL2009,BBM2008,BBM2011,EP2009}. However, these results are not permutations to our knowledge.  Certainly, we may transform them to permutations by CCZ equivalence. However, the research seems to be quite difficult and is not suitable to expand in the present paper. 

\subsection{$4$-uniform DDT permutations}

As is well-known, there are five classes of primarily constructed  $4$-uniform DDT permutations over $\gf_{2^n}$, which are listed in Table II. In Table II, ``some conditions" for Bracken-Tan-Tan function  refer to that  $ n=3k $, $k$ is even, $3\nmid k$, $k/2$ is odd, $\gcd(3k,s)=2$, $3 \mid {k+s}$ and $\alpha$ is a primitive element of $\gf_{2^n}$. 

\begin{table}[!htbp]
	\label{4-uniform-mon}
	\caption{$4$-uniform DDT permutations over $\gf_{2^n}$}
	\centering
\begin{tabular}{cccc}	
	\toprule
	Function & Expression & Conditions & Ref. \\
	\midrule
 Gold &	$x^{2^i+1}$ & $n=2k, k$ odd, $\gcd(n,i)=2$  &  \cite{Gold1968} \\
 Kasami &	$x^{2^{2i}-2^i+1}$ & $n=2k, k$ odd, $\gcd(n,i)=2$ & \cite{Kasami1971}  \\
 Inverse &	$x^{-1}$ & $n$ even & \cite{BD1993,Nyberg1994} \\
Bracken-Leander &	$x^{2^{2k}+2^k+1}$ & $n=4k$, $k$ odd & \cite{BL2010,Dobbertin1998} \\
Bracken-Tan-Tan & $\alpha x^{2^s+1}+\alpha^{2^k}x^{2^{-k}+2^{k+s}}$ & some conditions &  \cite{BTT2012} \\
\bottomrule
\end{tabular}	
\end{table}

In the subsection, we mainly consider boomerang uniformities of these permutation monomials over $\gf_{2^n}$ listed in Table II. 
\mkq{In fact, the boomerang uniformities of Gold and Inverse functions have been determined in \cite[Proposition 8]{BC2018}  and  \cite[Proposition 6]{BC2018}, respectively.  	As for Kasami, Bracken-Leander and Bracken-Tan-Tan functions, we list boomerang uniformities of these functions in small finite fields.  }
	\begin{table}[!htbp]
		\label{Kasami}
		\caption{The boomerang uniformity of the  Kasami function over $\gf_{2^{2k}}$}
		\centering
		\begin{tabular}{ccc|ccc}	
			\toprule
			Conditions &  Functions & Uniformities & Conditions &  Functions & Uniformities \\
			\midrule
		   $k=3, i=2$ & $x^{13}$ & $4$ &  $k=5,i=6$ & $x^{4033}$ & $44$  \\ 
		   $k=3, i=4$ & $x^{241}$ & $4$ &  $k=7, i=2$ & $x^{13}$ & $24$ \\
		   $k=5, i=2$ & $x^{13}$ & $44$ &   $k=7, i=4$ & $x^{241}$ & $16$ \\
		   $k=5, i=4$ & $x^{241}$ & $44$ & $k=7,i=6$ & $x^{4033}$ & $16$ \\
			\bottomrule
		\end{tabular}	
	\end{table}

   	\begin{table}[!htbp]
   	\label{Bracken-Leander}
   	\caption{The boomerang uniformity of the  Bracken-Leander function  over $\gf_{2^{4k}}$}
   	\centering
   	\begin{tabular}{ccc}	
   		\toprule
   		Conditions &  Functions & Uniformities  \\
   		\midrule
   		$k=1$ & $x^{7}$ & $4$   \\ 
   		$k=3$ & $x^{73}$ & $14$  \\
   		\bottomrule
   	\end{tabular}	
   \end{table}

{	From Tables III and IV, we can see that boomerang uniformities of Kasami and  Bracken-Leander functions become very high as the value of $k$ increases and it is the reason why we do not give theoretical result about boomerang uniformities of those two classes of functions. As for Bracken-Tan-Tan function \cite{BTT2012}: $ f(x) = \alpha x^{2^s+1}+\alpha^{2^k}x^{2^{-k}+2^{k+s}} $ over $\gf_{2^n}$, where $ n=3k $, $k$ is even, $3\nmid k$, $k/2$ is odd, $\gcd(3k,s)=2$, $3 \mid {k+s}$ and $\alpha$ is a primitive element of $\gf_{2^n}$, when $k = 2$, we have $s\equiv4\pmod6$ and $f(x)=\left(\alpha+\alpha^4\right) x^{15}$, which is one case of Gold functions and whose boomerang uniformity is $4$. Furthermore, when $k$ is bigger, like $k = 10$,  we can not compute the boomerang uniformity of $f(x)$ by Personal Computer within an acceptable time. However, we can see that the boomerang uniformity of  Bracken-Tan-Tan function is smaller than $12$ for any $k$ satisfying conditions according to Proposition \ref{Quadratic}. }

\subsection{ $4$-uniform DDT permutations constructed from the inverse function}

Recently, there were some $4$-uniform DDT permutations constructed from the inverse function, like \cite{LWY2013,PT2017,QTTL2013,TCT2015,QTLG2016} and the references therein. In this section, we mainly consider the function \mkq{over $\gf_{2^n}$}

 \begin{equation}
 \label{inv-switch}
f(x)=
\left\{
\begin{aligned}
1, & ~~ \text{if}~~ x = 0, \\
0, & ~~ \text{if}~~  x = 1,  \\
\frac{1}{x}, & ~~ \text{otherwise.}~~
\end{aligned}
\right.
\end{equation}
In \cite{LWY2013}, the authors proved that the differential uniformity of $f(x)$ over $\gf_{2^n}$ defined by (\ref{inv-switch}) is at most equal to $6$ and it is equal to $4$ if and only if $n\equiv2\pmod4$. Furthermore, $f(x)$ is with the best known nonlinearity. In the following, we consider the boomerang uniformity  of $f(x)$, i.e., the maximum of numbers of solutions in $\gf_{2^n}\times \gf_{2^n}$ of the following equation system for any $a,b\in\gf_{2^n}^{*}$,
\begin{equation}
	\label{BCT-uniform}
\left\{
\begin{array}{lr}
	f(x+a) + f(y+a) = b,  \\
f(x) + f(y) = b.
\end{array}
\right.
\end{equation}

\begin{Lemma}
	\label{lem-sparse-point}
	Let $f(x)$ be defined by (\ref{inv-switch}) and $a\notin \{1,\omega,\omega^2\}$. Then 
	\begin{enumerate}[(1)]
		\item $(0,y)$ and $(y,0)$ are two solutions of (\ref{BCT-uniform}) if and only if
			  $b=\frac{1+a}{a}$ or $a^2b^2+a^2b+ab+1=0$. In the case, $y=\frac{1}{b+1}$.
		\item $(1,y)$ and $(y,1)$ are two solutions of (\ref{BCT-uniform}) if and only if
 $b=\frac{1}{1+a}$ or 	$a^2b^2+ab^2+ab+1=0$. In the case, $y = \frac{1}{b}$.
	  \item $(a,y)$ and $(y,a)$ are two solutions of (\ref{BCT-uniform}) if and only if $a^2b^2+a^2b+ab+1=0$. In the case, $y=\frac{ab+a+1}{b+1}$;
	  \item $(a+1,y)$ and $(y,a+1)$ are two solutions of (\ref{BCT-uniform}) if and only if  $a^2b^2+ab^2+ab+1=0$. In the case, $y=\frac{ab+1}{b}$.
	\end{enumerate}
\end{Lemma}

\begin{proof}
	(1) If $(0,y)$  is a solution of (\ref{BCT-uniform}), then we have
\begin{subequations} 	\label{BCT-uniform-(1)}
	\renewcommand\theequation{\theparentequation.\arabic{equation}}     
	\begin{empheq}[left={\empheqlbrace\,}]{align}
		\frac{1}{a} &+ f(y+a) = b \label{BCT-uniform-(1):1A} \\ 
1 &+ f(y) = b.  \label{BCT-uniform-(1):1B} 
	\end{empheq}
\end{subequations}
	Thus $f(y)=b+1$.  
	
 If $b=1$, then $y=1$ and $\frac{1}{a}+\frac{1}{1+a}=1,$ leading to $a=\omega$ or $\omega^2$. Contradictions! Hence in the case $b\neq1$ and $y=\frac{1}{b+1}$.

 In the following, we let $b\neq1$. 
	 If $y=a$ holds at the same time,  $b=\frac{1+a}{a}$. On the other hand, when  $b=\frac{1+a}{a}$, it is easy to check that $(0,a)$ and $(a,0)$ do satisfy (\ref{BCT-uniform-(1)}).  If $y=a+1$ in the meantime, $a+1=\frac{1}{b+1}$. Moreover, $\frac{1}{a}=b$ holds according to (\ref{BCT-uniform-(1):1A}). Therefore, $a=b=1$, which is a contradiction.  If $y\neq a,a+1$, we have $\frac{1}{a}+\frac{1}{y+a}=b$ from  (\ref{BCT-uniform-(1):1A}). Plugging $y=\frac{1}{b+1}$ into $\frac{1}{a}+\frac{1}{y+a}=b$, we obtain $a^2b^2+a^2b+ab+1=0$. Furthermore, when $a^2b^2+a^2b+ab+1=0$ holds, we can also check that $(0,\frac{1}{b+1}), (\frac{1}{b+1},0)$ are two solutions of (\ref{BCT-uniform-(1)}).
	
	(2) If $(1,y)$ is a solution of (\ref{BCT-uniform}), we have 
\begin{subequations} 	\label{BCT-uniform-(2)}
	\renewcommand\theequation{\theparentequation.\arabic{equation}}     
	\begin{empheq}[left={\empheqlbrace\,}]{align}
	\frac{1}{a+1} &+ f(y+a) = b \label{BCT-uniform-(2):1A} \\ 
& f(y) = b.  \label{BCT-uniform-(2):1B} 
	\end{empheq}
\end{subequations}
  
   If $b=1$, from (\ref{BCT-uniform-(2)}), we have $y=0$ and $\frac{1}{1+a}+\frac{1}{a}=1,$ which is a contradiction. Thus $b\neq1$ and $y=\frac{1}{b}$ from  (\ref{BCT-uniform-(2):1A}). 
  
   In the following, we let $b\neq1$.
   If $y=a$, we have $\frac{1}{a+1}=b+1$ from (\ref{BCT-uniform-(2):1A}). Together with $\frac{1}{b}=a$ and $\frac{1}{a+1}=b+1$, $a=1$, which is contradictory. 
   If $y=a+1$, we have $b=\frac{1}{a+1}$. Furthermore, when $b=\frac{1}{a+1}$, $(1,a+1), (a+1,1)$ are two solutions of (\ref{BCT-uniform-(2)}). 
	If $y\neq a,a+1$, we have $\frac{1}{1+a}+\frac{1}{y+a}=b$ from (\ref{BCT-uniform-(2):1A}). Plugging $y=\frac{1}{b}$ into  $\frac{1}{1+a}+\frac{1}{y+a}=b$, we get $a^2b^2+ab^2+ab+1=0$. Moreover, when $a^2b^2+ab^2+ab+1=0$ holds, $(0,\frac{1}{b})$, $(\frac{1}{b},0)$ are two solutions of (\ref{BCT-uniform-(2)}).

 (3) If $(a,y)$ is a solution of (\ref{BCT-uniform}), we have 
\begin{subequations} \label{BCT-uniform-(3)}
	\renewcommand\theequation{\theparentequation.\arabic{equation}}     
	\begin{empheq}[left={\empheqlbrace\,}]{align}
 	1 & + f(y+a) = b \label{BCT-uniform-(3):1A} \\ 
&	\frac{1}{a}  + f(y) = b.  \label{BCT-uniform-(3):1B} 
	\end{empheq}
\end{subequations}

If $b=1$, from (\ref{BCT-uniform-(3)}), we have $y=a+1$ and $\frac{1}{a}+\frac{1}{a+1}=1$, which is a contradiction. Thus $b\neq1$ and $y=\frac{1}{b+1}+a=\frac{ab+a+1}{b+1}.$ Together with (\ref{BCT-uniform-(3):1B}), we get $$b+\frac{1}{a}=f(y)=\frac{b+1}{ab+a+1},$$ i.e., 
$$a^2b^2+a^2b+ab+1=0.$$ Moreover, when $a^2b^2+a^2b+ab+1=0$, $(a,\frac{ab+a+1}{b+1})$ and $(\frac{ab+a+1}{b+1},a)$ are two solutions of (\ref{BCT-uniform-(3)}). 

 (3) If $(a+1, y)$ is a solution of (\ref{BCT-uniform}), we have 
\begin{subequations} 	\label{BCT-uniform-(4)}
	\renewcommand\theequation{\theparentequation.\arabic{equation}}     
	\begin{empheq}[left={\empheqlbrace\,}]{align}
	   f(y+a) & = b \label{BCT-uniform-(4):1A} \\ 
\frac{1}{a+1}  + f(y)	& = b.  \label{BCT-uniform-(4):1B} 
	\end{empheq}
\end{subequations}

If $b=1$, from (\ref{BCT-uniform-(4)}), we have $y=a$ and $\frac{1}{a}+\frac{1}{a+1}=1$, which is a contradiction. Thus $b\neq1$ and $y=\frac{1}{b}+a=\frac{ab+1}{b}.$ Together with  (\ref{BCT-uniform-(4):1B}), we get $$b+\frac{1}{a+1}=f(y)=\frac{ab+1}{b},$$ i.e., 
$$a^2b^2+ab^2+ab+1=0.$$ Moreover, when $a^2b^2+ab^2+ab+1=0$, $(a+1,\frac{ab+1}{b})$ and $(\frac{ab+1}{b},a+1)$ are two solutions of (\ref{BCT-uniform-(4)}). 
\end{proof}

Let sets $S_i ( i = 1, \cdots, 4)$ be the conditions of Lemma \ref{lem-sparse-point}, respectively, i.e.,
$$S_1 :=\{ (a,b) \in\gf_{2^n}^{*}\times\gf_{2^n}^{*} | b=\frac{1+a}{a}~~ \text{or}~~ a^2b^2+a^2b+ab+1=0 \}; $$
$$S_2 := \{ (a,b) \in\gf_{2^n}^{*}\times\gf_{2^n}^{*} | b=\frac{1}{1+a} ~~ \text{or}~~ a^2b^2+ab^2+ab+1=0\}; $$
$$S_3 := \{ (a,b) \in\gf_{2^n}^{*}\times\gf_{2^n}^{*} |  a^2b^2+a^2b+ab+1=0 \};$$
$$S_4 := \{ (a,b) \in\gf_{2^n}^{*}\times\gf_{2^n}^{*} | a^2b^2+ab^2+ab+1=0  \}. $$
 What should be noticed is that $a\notin \{1,\omega,\omega^2\}$ for $S_i ( i = 1, \cdots, 4)$ and there are at least two solutions of  (\ref{BCT-uniform}) if $(a,b)\in S_i$, where $i=1,2,3,4$. Moreover, if $(a,b)$ belongs to some $S_i$ at the same time, there are more solutions of (\ref{BCT-uniform}). Therefore, it is worthwhile to consider the intersections of $S_i$ and the following lemma can answer the problem, whose proof is omitted since it is easy to prove.
 
 \begin{Lemma}
 	\label{lem-S1S2S3S4}
 	\begin{enumerate}
 		\item $S_1\cap S_2 = \{ (a,b) \in\gf_{2^n}^{*}\times\gf_{2^n}^{*} | a^3+a+1=0 ~~\text{and}~~ b=a,\frac{1+a}{a},\frac{1}{a+1}  \}$;
 		\item $S_1\cap S_3 = S_3$;
 		\item $S_1\cap S_4 = \{ (a,b) \in\gf_{2^n}^{*}\times\gf_{2^n}^{*} | a^3+a+1=0 ~~\text{and}~~ b=a,\frac{1+a}{a}  \}$;
 		\item $S_2\cap S_3 = \{ (a,b) \in\gf_{2^n}^{*}\times\gf_{2^n}^{*} | a^3+a+1=0  ~~\text{and}~~ b=a,\frac{1}{a+1} \}$;
 		\item $S_2\cap S_4 = S_4$;
 		\item $S_3\cap S_4 = \{ (a,b)\in\gf_{2^n}^{*}\times\gf_{2^n}^{*} | a^3+a+1=0 ~~\text{and}~~ b=a \}.$
 		\item $S_1\cap S_2 \cap S_3 \cap S_4 = \{ (a,b) | a^3+a+1=0 ~~\text{and}~~ b=a \}. $
 	\end{enumerate}
 \end{Lemma}
  
\begin{Rem}
	\label{rem-8-4}
\emph{	It is clear that $a^3+a+1=0$ has solutions in $\gf_{2^n}$ if and only if $n\equiv0\pmod3$. Thus when $n\equiv 0\pmod3$, there exist some $a,b$ such that $a^3+a+1=0$ and $b=a$. Together with   Lemmas \ref{lem-sparse-point} and \ref{lem-S1S2S3S4}, for such $a, b$ satisfying $a^3+a+1=0$ and $b=a$, (\ref{BCT-uniform}) has at least eight solutions in $\gf_{2^n}\times\gf_{2^n}$. And when $n\not\equiv0\pmod3$, there exist some $a,b\in\gf_{2^n}$ satisfying $a^2b^2+a^2b+ab+1=0$ or $a^2b^2+ab^2+ab+1=0$, when which (\ref{BCT-uniform}) has at least four solutions in $\gf_{2^n}\times\gf_{2^n}$.}
\end{Rem}  
  
  In the following, we give the main theorem of this section, determining the boomerang uniformity of $f(x)$ defined by (\ref{inv-switch}). However, the proof is extremely tedious and we only introduce the prime idea here. The detailed and complete proof can be found in the Appendix.
  
\begin{Th}
	\label{BCT-inv-switch}
	Let $f(x)$ be defined by (\ref{inv-switch}) and $n\ge3$. Then the boomerang uniformity of $f$ is 
	\begin{equation*}
\delta_{f} =	\left\{
	\begin{aligned}
	10,  ~~\text{if}~~ n\equiv0\pmod6, \\
	8, ~~\text{if}~~ n\equiv 3 \pmod 6, \\
	6,   ~~\text{if}~~ n\not\equiv0\pmod3.
	\end{aligned}
	\right.
	\end{equation*}	 
\end{Th}  

\begin{proof} 
		It suffices to compute  $\max \limits_{a,b\in\gf_{2^n}^{*}} T(a,b)$, where $T(a,b)$ is the number of solutions in $\gf_{2^n}\times \gf_{2^n}$ of the following equation system for $a,b\in\gf_{2^n}\backslash\{0\}$,
		\begin{equation*}
		\left\{
		\begin{array}{lr}
		f(x+a) + f(y+a) = b,  \\
		f(x) + f(y) = b.
		\end{array}
		\right.
		\end{equation*}

We divide the problem into three cases:
\begin{enumerate}[(1)]
	\item  $a=1$, $b\in\gf_{2^n}^{*}$;
	\item $a=\omega$ or $\omega^2$, $b\in\gf_{2^n}^{*}$;
	\item $a\in\gf_{2^n}^{*}\backslash\{1,\omega,\omega^2 \}$,  $b\in\gf_{2^n}^{*}$. 
\end{enumerate}
For the details proof of each case, please see the Appendix. We only mention that Lemmas \ref{lem-sparse-point} and \ref{lem-S1S2S3S4} play important roles in cases (3) and for each case, the maximum of numbers of solutions in $\gf_{2^n}\times \gf_{2^n}$ of the above equation system is listed here:
   	\begin{table}[!htbp]
	\centering
	\begin{tabular}{cc}	
		\toprule
		Cases &  $\max \limits_{a,b\in\gf_{2^n}^{*}} T(a,b)$   \\
		\midrule
		$a = 1$ & $2$ if $n\equiv1\pmod2$, $4$ if $n\equiv2\pmod4$, $6$ if $n\equiv0\pmod4$ \\ 
		$a = \omega$ or $\omega^2$ &   $4$ if $n\equiv2\pmod4$, $6$ if $n\equiv0\pmod4$  \\
		$a \in \gf_{2^n}^{*}\backslash\{1,\omega,\omega^2 \} $ &  $10$ if $n\equiv0\pmod6$, $8$ if $n\equiv3\pmod6$, $6$ if $n\not\equiv0\pmod3$\\
		\bottomrule
	\end{tabular}	
\end{table}

Therefore, the boomerang uniformity of $f$ is as claimed.
\end{proof}

\begin{example}
\emph{	In order to test the correctness of Theorem \ref{BCT-inv-switch}, we compute the boomerang uniformity $\delta_{f}$ of $f$ defined by (\ref{inv-switch}) when $3\le n\le 9$ by Magma and the results are listed as follows.
	 \begin{table}[!htbp]
	 	\centering
	 	\begin{tabular}{cccccccc}	
	 		\toprule
	 		$n$ & $3$ & $4$ & $5$ & $6$ & $7$ & $8$ & $9$ \\
	 		\midrule
	 		$\delta_{f}$ & $8$  & $6$ & $6$  &  $10$ & $6$ & $6$ & $8$ \\
	 		\bottomrule
	 	\end{tabular}	
	 \end{table}}
\end{example}
  
From the proof of Theorem \ref{BCT-inv-switch}, we see that computing  boomerang uniformity is generally more complicated than computing differential uniformity and  it seems quite complicated to obtain the boomerang uniformity of all $4$-uniform DDT permutations constructed from inverse function.
  
\section{A class of $4$-uniform BCT permutations}

 In \cite{BCT2018}, Cid et al.  discussed that obtaining $4$-uniform BCT S-boxes appears to be hard, especially as the size of the S-box increases. In the following, we present a class of $4$-uniform BCT permutation polynomials over $\gf_{2^n}$.

In \cite{Zieve2013}, Zieve obtained some classes of permutation polynomials with the form of $x^rh\left(x^{q+1}\right)\in\gf_{q^2}[x]$, where $q$ is an arbitrary prime power by using all low-degree ($\le 5$) permutation polynomials over $\gf_q$ (cf. \cite[P352, Table 7.1]{LN1997}). The following theorem is one of these permutation polynomials and we can prove that it is a $4$-uniform BCT function (Theorem \ref{xq+2+gamx}) when $q$ is even. By the way, the permutation polynomial in Theorem \ref{xq+2+gamx} is also a $4$-uniform DDT function  which has been showed in \cite{ZZC2015}.

\begin{Lemma}
	\label{Zieveq+13}
	\cite{Zieve2013}
	Pick $\gamma\in\gf_{q^2}^{*}$, and write $f(x)=x^{q+2}+\gamma x.$ Then $f$ permutes $\gf_{q^2}$ if and only if one of the following occurs:
	\begin{enumerate}[(1)]
		\item $q\equiv5\pmod6$ and $\gamma^{q-1}$ has order $6$;
		\item $q\equiv2\pmod6$ and $\gamma^{q-1}$ has order $3$; or
		\item $q\equiv0\pmod3$ and $\gamma^{q-1}=-1.$
	\end{enumerate}
\end{Lemma}

 Leonard and  Williams characterized the factorization of a quartic polynomial over $\gf_{2^n}$ in \cite{LW1972} and the result is useful to compute the boomerang uniformity of $f(x)$ in Theorem \ref{xq+2+gamx}.

\begin{Lemma}
	\cite{LW1972}
	\label{lem-quartics}
	Let $q=2^n$ and $f(x)=x^4+a_2x^2+a_1x+a_0\in\gf_q[x]$, where $a_0a_1\neq0$. Let $g(x)=x^3+a_2x+a_1$ and $r_1,r_2,r_3$ be three roots in $\gf_q$ if they exist in $\gf_q$. Then $f(x)$ has four solutions in $\gf_q$ if and only if $g(y)$ has three solutions in $\gf_{q}$ and $\tr_{q}\left(\frac{a_0r_1^2}{a_1^2}\right)=\tr_{q}\left(\frac{a_0r_3^2}{a_1^2}\right)=\tr_{q}\left(\frac{a_0r_3^2}{a_1^2}\right)=0.$
\end{Lemma}

\begin{Th}
	\label{xq+2+gamx}
	Let $q=2^n$, $n$ be odd and $f(x)=x^{q+2}+\gamma x\in\gf_{q^2}[x]$, where $\mathrm{Ord}\left(\gamma^{q-1}\right)=3$. Then $\delta_f=4$.
\end{Th}

\begin{proof}
	It suffices to prove that for $a,b\in\gf_{q^2}^{*}$, the equation system 
\begin{equation*}
\left\{
\begin{array}{lr}
f(x+a) + f(y+a) = b,  \\
f(x) + f(y) = b.
\end{array}
\right.
\end{equation*}
  i.e., 
 \begin{equation*}
 \left\{
 \begin{array}{lr}
 f(x+a) + f(y+a) =  f(x) + f(y),  \\
 f(x) + f(y) = b.
 \end{array}
 \right.
 \end{equation*}
   has at most $4$ solutions in $\gf_{q^2}\times \gf_{q^2}$. Thanks to $$f(x+a)=(x+a)^{q+2}+\gamma(x+a)=x^{q+2}+a^2x^q+a^qx^2+a^{q+2}+\gamma x+\gamma a,$$  simplifying the above equation system, we have 
\begin{subequations} 		\label{xq+2+x-BCT}
	\renewcommand\theequation{\theparentequation.\arabic{equation}}     
	\begin{empheq}[left={\empheqlbrace\,}]{align}
 a^2(x+y)^q+a^q(x+y)^2=0 \label{xq+2+x-BCT:1A} \\ 
x^{q+2}+\gamma x +y^{q+2} + \gamma y = b. \label{xq+2+x-BCT:1B} 
	\end{empheq}
\end{subequations}

   Let $z=x+y$. Then from (\ref{xq+2+x-BCT:1A}), we obtain $z=\beta a$, where $\beta =1, \omega$ or $\omega^2$, since $\gcd\left(q-2,q^2-1\right)=3$ when $n$ is odd.  
   Moreover, $y=x+z=x+\beta a$. Plugging it into (\ref{xq+2+x-BCT:1B}), we have 
   \begin{equation*}
   \beta^2a^2x^q+a^q\beta^2x^2+ \beta a^{q+2}+ \gamma \beta a +b=0.
   \end{equation*}
   Let $x=aX$. Then the above equation becomes 
     \begin{equation}
   \label{xq+2+x-BCT-x}
   X^q+X^2=c,
   \end{equation}
   where $c=\frac{\beta a^{q+2}+\gamma \beta a + b}{\beta^2 a^{q+2}}.$  
   Let $L(x)= x^q + x^2$. Since $L(x)=0$ has $4$ solutions in $\gf_{q^2}$, $L(x)$ is a $4$-to-$1$ polynomial over $\gf_{q^2}$. Thus Eq. (\ref{xq+2+x-BCT-x}) has $4$ or $0$ solutions in $\gf_{q^2}$ for a given $\beta$. In the following, we show that given any $a,b\in\gf_{q^2}^{*}$, for three cases $\beta= 1, \omega, \omega^2$, Eq. (\ref{xq+2+x-BCT-x}) can not have solutions at the same time.   First of all, we show that equations $x^q+x^2=c$ and $x^4+x=c$ have solutions or no solutions in $\gf_{q^2}$ at the same time. On one side, if $x_0\in\gf_{q^2}$ is a solution of $x^4+x=c$, let $x_1 = x_0^q+x_0^2$. Then $x_1^q+x_1^2 = x_0+x_0^{2q}+x_0^{2q}+x_0^4 = c$, which means that $x_1$ is a solution of $x^q+x^2=c$. On the other side, if $x_0$ is a solution of $x^q+x^2=c$, let $x_1 = x_0^{2^{n-2}} + x_0^{2^{n-4}} + \cdots + x_0^2$. Then $x_1^4+x_1 = x_0^{2^n}+x_0^2 = c$, claiming that $x_1$ is a solution of $x^4+x=c$. Therefore, given any  $a,b\in\gf_{q^2}^{*}$,  we suffice to consider equation 
   \begin{equation}
   \label{x4+x=c}
   x^4+x=c
   \end{equation}
   can not have solutions at the same time in $\gf_{q^2}$ for three cases $\beta= 1, \omega, \omega^2$, where $c=\frac{\beta a^{q+2}+\gamma \beta a + b}{\beta^2 a^{q+2}}.$ Let $c_1=\frac{a^{q+2}+\gamma a+b}{a^{q+2}}$, $c_2=\frac{\omega a^{q+2}+\gamma \omega a+b}{\omega^2a^{q+2}}$ and $c_3=\frac{\omega^2a^{q+2}+\gamma \omega^2 a+b}{\omega a^{q+2}}.$
    
   Let $g(x)=x^3+1$. Then it is trivial that $g(x)$ has three roots $x = 1, \omega, \omega^2$ in $\gf_{q^2}$. Moreover, according to Lemma \ref{lem-quartics}, Eq. (\ref{x4+x=c}) has four solutions in $\gf_{q^2}$ if and only if $\tr_{q^2}(c)=\tr_{q^2}(\omega c)=\tr_{q^2}(\omega^2 c)=0$. Without loss of generality, we assume that equations $x^4+x=c_1$ and $x^4+x=c_2$ has four solutions in the meantime. Thus we have $\tr_{q^2}\left(c_1\right)=\tr_{q^2}\left(\omega c_1\right)=\tr_{q^2}\left(c_2\right)=\tr_{q^2}\left(\omega c_2\right)=0$. In addition, $\tr_{q^2}(c)=0$ if and only if there exist some $z\in\gf_{q^2}$ such that $c=z+z^2$. Therefore, there exist $z_1, z_2, z_3, z_4\in\gf_{q^2}$ such that $c_1 = z_1 + z_1^2$, $\omega c_1 = z_2 + z_2^2$, $c_2 = z_3 + z_3^2$  and $\omega c_2 = z_4 + z_4^2$. Moreover, we have $\omega(z_2+z_2^2) = (z_1+z_2)+(z_1+z_2)^2$ and $\omega (z_3+z_3^2)=z_4+z_4^2$. Plugging $c_1=\frac{a^{q+2}+\gamma a+b}{a^{q+2}}$, $c_2=\frac{\omega a^{q+2}+\gamma \omega a+b}{\omega^2a^{q+2}}$ into $\omega c_1 = z_2 + z_2^2$ and $c_2 = z_3 + z_3^2$, we have 
   \begin{equation}
   \label{z2}
   \frac{\omega\gamma}{a^{q+1}}+\frac{\omega b}{a^{q+2}} = z_2+z_2^2+\omega
   \end{equation}
   and 
   \begin{equation}
   \label{z3}
   \frac{\omega^2 \gamma}{a^{q+1}} + \frac{\omega b}{a^{q+2}} =z_3+z_3^2+\omega^2.
   \end{equation}
   Adding Eq. (\ref{z2}) and Eq. (\ref{z3}), we obtain 
   \begin{equation}
   \label{z2z3}
   \frac{\gamma}{a^{q+1}}=z_2+z_2^2+z_3+z_3^2+1.
   \end{equation}
   Since  $\mathrm{Ord}\left(\gamma^{q-1}\right)=3$, $\gamma^{3q}=\gamma^3$. Furthermore, $\gamma^q=\beta \gamma$, where $\beta = \omega$ or $\omega^2$. In the following, we assume $\beta=\omega$ and the other case is similar.
   Raising  Eq. (\ref{z2z3}) into its $q$-th power, we have 
   \begin{equation}
   \label{z2z3q}
   \frac{\omega \gamma}{a^{q+1}}=z_2^q+z_2^{2q}+z_3^q+z_3^{2q}+1.
   \end{equation}
   Adding $\omega\times$(\ref{z2z3}) and (\ref{z2z3q}), we get 
   $$\omega\left(z_2+z_2^2+z_3+z_3^2+1\right)=z_2^q+z_2^{2q}+z_3^q+z_3^{2q}+1,$$
   i.e., 
   $$(z_1+z_2)+(z_1+z_2)^2+z_4+z_4^2+\omega=z_2^q+z_2^{2q}+z_3^q+z_3^{2q}+1.$$
   Thus $$\tr_{q^2}\left((z_1+z_2)+(z_1+z_2)^2+z_4+z_4^2+\omega\right)=\tr_{q^2}\left(z_2^q+z_2^{2q}+z_3^q+z_3^{2q}+1\right).$$
   Since $\tr_{q^2}(\omega)=1$ and $\tr_{q^2}(1)=0$, we have $1=0$ from the above equation. Contradictions! 
   
    Therefore, Eq. (\ref{xq+2+x-BCT}) has at most $4$ solutions in $\gf_{q^2}\times \gf_{q^2}$.  That is to say, $\delta_f=4.$

We have finished the proof.   
\end{proof}

It is well known that given a permutation polynomial over $\gf_q$, it is very hard to compute its explicit compositional inverse. In \cite{LQW2018}, the authors introduced an approach to compute the explicit expression for compositional inverses of permutation polynomials of the form $x^rh\left(x^s\right)$ over $\gf_q$, where $s\mid (q-1)$ and $\gcd(r,q-1)=1$. Their main idea is to transform the problem of computing the compositional inverses of permutation polynomials of the form $x^rh\left(x^s\right)$ into that of computing the compositional inverses of two restricted permutation mappings, i.e., $x^r$ over $\gf_q$ and $x^rh(x)^s$ over the set of $(q-1)/s$-th roots of unity in $\gf_q^{*}$. Furthermore, they computed the explicit compositional inverses of the permutation polynomials in Theorem \ref{Zieveq+13} for arbitrary $q$. 

\begin{Lemma}
	\label{inverse-xq+2+x}
	\cite{LQW2018}
	Let $f(x)=x^{q+2}+\gamma x\in\gf_{q^2}[x]$ be a permutation polynomial over $\gf_{q^2}$.
	\begin{enumerate}[(1)]
		\item If $q\equiv2\pmod3$, then the compositional inverse of $f(x)$ over $\gf_{q^2}$ is $$f^{-1}(x)=x^{q^2-q-1}\left(\left(x^{q+1}+\epsilon^3\right)^{2\cdot3^{-1}}-(2\epsilon-\gamma^q)\left(x^{q+1}+\epsilon^3\right)^{3^{-1}}+\epsilon^2-\epsilon\gamma^q\right),$$
		where $\epsilon=\frac{\gamma^q+\gamma}{3}$ and $3^{-1}=\frac{2q-1}{3}$ is the compositional inverse of $3$ modulo $q-1$.
		\item If $q\equiv0\pmod3$, i.e., $q=3^n$, then the compositional inverse of $f(x)$ over $\gf_{q^2}$ is
		$$f^{-1}(x)=-\left(\sum_{i=0}^{n-1}\gamma^{{-3^{i+1}+1}}x^{3^{i}(q+1)}+\gamma \right)\left(\sum_{i=0}^{n-1}\gamma^{{-3^{i+1}+1}}x^{3^{i}(q+1)-q}\right).$$
	\end{enumerate}
\end{Lemma}

Together with Theorem \ref{xq+2+gamx} and Lemma \ref{inverse-xq+2+x}, we can obtain the following result.

\begin{Cor}
		Let $q=2^n$, $n$ be odd and $$f(x)=x^{q^2-q-1}\left(\left(x^{q+1}+\epsilon^3\right)^{2\cdot3^{-1}}+\gamma^q\left(x^{q+1}+\epsilon^3\right)^{3^{-1}}+\epsilon^2+\epsilon\gamma^q\right)\in\gf_{q^2}[x],$$ where $\mathrm{Ord}\left(\gamma^{q-1}\right)=3$, $\epsilon=\gamma^q+\gamma$ and $3^{-1}=\frac{2q-1}{3}$ is the compositional inverse of $3$ modulo $q-1$. Then $\delta_f=4$.
\end{Cor}

\section{Conclusion}

\mkq{Boomerang Connectivity Table (BCT for short) is a new cryptanalysis tool introduced by Cid et al. \cite{BCT2018} in EUROCRYPT 2018, to  evaluate the subtleties of boomerang-style attacks. In this paper, we give some new properties about BCT and the boomerang uniformity. 
Firstly, we give an equivalent and simple {formula to compute}  BCT and the boomerang uniformity. The advantage of our new method is not only that the compositional inverse is not needed, but also that the definition of BCT and the boomerang uniformity can be generalized. Secondly, we give a characterization of $\delta_{f}$-uniform BCT functions by means of the Walsh transform. In particular, a new equivalent characterization about APN functions is presented. Thirdly, we consider boomerang uniformities of some special permutations with low differential uniformity. Finally, we obtain a new class of $4$-uniform BCT permutations over $\gf_{2^n}$, which is the first binomial. It is worth
mentioning that it seems not easy to compute the boomerang uniformity of this class of binomial from the original
definition directly .

From \cite{BC2018} and our results, we can see that there exist $2$-uniform BCT permutation over $\gf_{2^n}$, where $n$ is odd and $4$-uniform BCT permutations over $\gf_{2^n}$, where $n\equiv2\pmod4$. However, for the case $n\equiv0\pmod4$, which are very widely used in cryptographic algorithm, we can not find any permutations over $\gf_{2^n}$ with boomerang uniformity $4$ up to now and it is our next goal.}

\mkq{
	\bigskip\noindent
	{\bf Acknowledgment} \
	We would like to thank the editor and the anonymous referees whose valuable comments and suggestions improve both the technical quality and the editorial quality of this paper.}

\appendix
{\bfseries The whole proof of Theorem \ref{BCT-inv-switch}}
\begin{proof}
	It suffices to consider the numbers of solutions in $\gf_{2^n}\times \gf_{2^n}$ of the following equation system for $a,b\in\gf_{2^n}^{*}$,
\begin{equation}
\label{BCT-uniform-inv-switch}
\left\{
\begin{array}{lr}
f(x+a) + f(y+a) = b, \\
f(x) + f(y) = b.
\end{array}
\right.
\end{equation}

	{\bfseries Case 1: $a=1$.} 
	
	{\bfseries Subcase 1.1: $a=1, b=1$.} In the subcase, (\ref{BCT-uniform-inv-switch}) becomes 
\begin{equation}
	\label{BCT-uniform-inv-switch-1.1}
\left\{
\begin{array}{lr}
f(x+1) + f(y+1) = 1, \\
f(x) + f(y) = 1.
\end{array}
\right.
\end{equation}
	
	(i) If $x=0$ or $1$, $y=1$ or $0$, respectively. Thus $(x,y)=(0,1), (1,0)$ are two solutions of (\ref{BCT-uniform-inv-switch-1.1});
	
	(ii) If $x\neq 0, 1$, $y\neq 0, 1$, either. Then (\ref{BCT-uniform-inv-switch-1.1}) becomes 
\begin{equation}
	\label{BCT-uniform-inv-switch-1.1-1}
\left\{
\begin{array}{lr}
	\frac{1}{x+1} + \frac{1}{y+1} = 1, \\
\frac{1}{x}  + \frac{1}{y} = 1.
\end{array}
\right.
\end{equation}
	
	After simplifying (\ref{BCT-uniform-inv-switch-1.1-1}), we obtain $x^2+x+1=0$. Therefore, when $n\equiv1\pmod2$, (\ref{BCT-uniform-inv-switch-1.1-1}) has no solutions in $\gf_{2^n} \times \gf_{2^n}$. When $n\equiv0\pmod2$, (\ref{BCT-uniform-inv-switch-1.1-1}) has two solutions $(x,y)=(\omega,\omega^2)$ and $(\omega^2,\omega)$.
	
	Hence, in the subcase $a=1, b=1$, the number of solutions in $\gf_{2^n}\times\gf_{2^n}$ of (\ref{BCT-uniform-inv-switch}) is $2$ (when $n\equiv1\pmod2$) or $4$ (when $n\equiv0\pmod2$).
	
	{\bfseries Subcase 1.2: $a=1$, $b=\omega$ or $\omega^2$.} What should be noticed is that the subcase is under the case $n$ even. In the subcase, we only consider the case $b=\omega$ since the other case is similar. Obviously, (\ref{BCT-uniform-inv-switch}) becomes 
\begin{equation}
	\label{BCT-uniform-inv-switch-1.2}
\left\{
\begin{array}{lr}
	f(x+1) + f(y+1) = \omega, \\
f(x) + f(y) = \omega.
\end{array}
\right.
\end{equation}

	(i) If $x=0$, $f(y+1)=\omega$ and $f(y)=\omega+1=\omega^2$ from (\ref{BCT-uniform-inv-switch-1.2}), which means $y=\omega$. Similarly, when $x=1$, $y=\omega^2$. Thus $(0,\omega), (\omega,0), (1,\omega^2)$ and $(\omega^2,1)$ are four solutions of (\ref{BCT-uniform-inv-switch-1.2}).
	
	(ii) If $x\neq0,1$, (\ref{BCT-uniform-inv-switch-1.2}) becomes
\begin{equation}
	\label{BCT-uniform-inv-switch-1.2-1}
\left\{
\begin{array}{lr}
	\frac{1}{x+1} + \frac{1}{y+1} = \omega, \\
\frac{1}{x}  + \frac{1}{y} =  \omega.
\end{array}
\right.
\end{equation}

	After simplifying (\ref{BCT-uniform-inv-switch-1.2-1}), we have $x^2+x+\omega^2=0$, which has two solutions $x_0,x_0+1$ in $\gf_{2^n}$ if and only if $n\equiv0\pmod4$. Moreover, when $n\equiv0\pmod4$, $(x_0,x_0+1)$ and $(x_0+1,x_0)$ are two solutions in $\gf_{2^n}\times \gf_{2^n}$ of (\ref{BCT-uniform-inv-switch-1.2-1}). 
	
	Hence, in the subcase $n\equiv0\pmod2$, $a=1, b=\omega$ or $\omega^2$, the number of solutions in $\gf_{2^n}\times\gf_{2^n}$ of (\ref{BCT-uniform-inv-switch}) is $4$ (when $n\equiv2\pmod4$) or $6$ (when $n\equiv0\pmod4$).
	
	{\bfseries Subcase 1.3: $a=1$, $b\neq\{1,\omega,\omega^2\}$.  } In the subcase, (\ref{BCT-uniform-inv-switch}) becomes 
\begin{equation}
	\label{BCT-uniform-inv-switch-1.3}
\left\{
\begin{array}{lr}
	f(x+1) + f(y+1) = b, \\
f(x) + f(y) = b.
\end{array}
\right.
\end{equation}

	(i) If $x=0$, $f(y)=b+1$ and $f(y+1)=b$. Furthermore, $y=\frac{1}{b+1}$ and $y=\frac{1}{b}+1$, which means $\frac{1}{b+1}=\frac{1}{b}+1$ and $b=\omega, \omega^2$. Contradictions! Thus $x\neq0$. Similarly, we also can obtain $x\neq1$, $y\neq0,1$. 
	
	(ii) If $x\neq0,1$ and $y\neq0,1$, we have 
\begin{equation}
\label{BCT-uniform-inv-switch-1.3-1}
\left\{
\begin{array}{lr}
\frac{1}{x+1} + \frac{1}{y+1} = b , \\
\frac{1}{x}  + \frac{1}{y} =  b.
\end{array}
\right.
\end{equation}

	From the above equation system, we get $bx^2+bx+1=0$, which at most two solutions in $\gf_{2^n}$. Moreover, in the subcase, (\ref{BCT-uniform-inv-switch}) has at most two solutions in $\gf_{2^n}\times\gf_{2^n}$. 
	
	Therefore, in the case $a=1$,  the maximum of numbers of solutions in $\gf_{2^n}\times\gf_{2^n}$ of (\ref{BCT-uniform-inv-switch}) is
	\begin{equation*}
	\left\{
	\begin{aligned}
	2,  ~~\text{if}~~ n\equiv1\pmod2, \\
	4,  ~~\text{if}~~ n\equiv2\pmod4, \\
	6,   ~~\text{if}~~ n\equiv0\pmod4.
	\end{aligned}
	\right.
	\end{equation*}
	
	{\bfseries Case 2: $a=\omega, \omega^2$.} In the case, we should notice that $n\equiv0\pmod2$ and only consider $a=\omega$ since the other case $a=\omega^2$ is similar.  
	
	{\bfseries Subcase 2.1:  $b=1$. } In the subcase, (\ref{BCT-uniform-inv-switch}) becomes 
\begin{equation}
	\label{BCT-uniform-inv-switch-2.1}
\left\{
\begin{array}{lr}
	f(x+\omega) + f(y+\omega) = 1, \\
f(x) + f(y) = 1. 
\end{array}
\right.
\end{equation}
%

	(i) If 	$x=0, 1$, $y=1, 0$ respectively. Thus $(0,1)$ and $(1,0)$ are two solutions of (\ref{BCT-uniform-inv-switch-2.1}).
	
	(ii) If $x=\omega, \omega^2$, $y=\omega^2, \omega$ respectively. Thus $(\omega, \omega^2)$ and $(\omega^2, \omega)$ are two solutions of (\ref{BCT-uniform-inv-switch-2.1}).
	
	(ii) If $x\neq 0, 1, \omega, \omega^2$, (\ref{BCT-uniform-inv-switch-2.1}) becomes 
\begin{equation}
\label{BCT-uniform-inv-switch-2.1-1}
\left\{
\begin{array}{lr}
\frac{1}{x+\omega} + \frac{1}{y+\omega} = 1, \\
\frac{1}{x}  + \frac{1}{y} = 1. 
\end{array}
\right.
\end{equation}

	After simplifying (\ref{BCT-uniform-inv-switch-2.1-1}), we have $y=\frac{x}{x+1}$ and $x^2+\omega x+\omega=0$, which has two solutions $x_0,x_0+\omega$ in $\gf_{2^n}$ if and only if $n\equiv0\pmod4$. Moreover, when $n\equiv0\pmod4$, due to $\frac{x_0}{x_0+1}=x_0+\omega$, $(x_0,x_0+\omega)$ and $(x_0+\omega,x_0)$ are two solutions in $\gf_{2^n}\times \gf_{2^n}$ of (\ref{BCT-uniform-inv-switch-2.1-1}). 
	
	Hence, in the subcase $n\equiv0\pmod2$, $a=\omega, b=1$, the number of solutions in $\gf_{2^n}\times\gf_{2^n}$ of (\ref{BCT-uniform-inv-switch}) is $4$ (when $n\equiv2\pmod4$) or $6$ (when $n\equiv0\pmod4$). 
	
	{\bfseries Subcase 2.2: $b=\omega$.}	 In the subcase, (\ref{BCT-uniform-inv-switch}) becomes
\begin{equation}
	\label{BCT-uniform-inv-switch-2.2}
\left\{
\begin{array}{lr}
f(x+\omega) + f(y+\omega) = \omega, \\
f(x) + f(y) = \omega. 
\end{array}
\right.
\end{equation}
	
	(i) If $x=0, 1, \omega, \omega^2$, $y= \omega, \omega^2, 0, 1$ respectively. Thus $(0,\omega), (\omega,0), (1, \omega^2) $ and $(\omega^2,1)$ are four solutions of (\ref{BCT-uniform-inv-switch-2.2}).
	
	(ii) If $x\neq 0, 1, \omega, \omega^2$, we have 
\begin{equation}
\label{BCT-uniform-inv-switch-2.2-1}
\left\{
\begin{array}{lr}
	\frac{1}{x+\omega} + \frac{1}{y+\omega} = \omega, \\
\frac{1}{x}  + \frac{1}{y} = \omega. 
\end{array}
\right.
\end{equation}

	After simplifying (\ref{BCT-uniform-inv-switch-2.2-1}), we have $y=\frac{x}{\omega x+1}$ and $x^2+\omega x+1=0$, which has two solutions $x_0, x_0+\omega$ if and only if $n\equiv0\pmod4$. Furthermore, when $n\equiv0\pmod4$, $(x_0,x_0+\omega), (x_0+\omega, x_0) $ are two solutions of (\ref{BCT-uniform-inv-switch-2.2-1}) and when $n\equiv2\pmod4$, (\ref{BCT-uniform-inv-switch-2.2-1}) has no solutions in $\gf_{2^n}\times\gf_{2^n}$.   
	
	Hence, in the subcase $n\equiv0\pmod2$, $a=\omega, b=\omega$, the number of solutions in $\gf_{2^n}\times\gf_{2^n}$ of (\ref{BCT-uniform-inv-switch}) is $4$ (when $n\equiv2\pmod4$) or $6$ (when $n\equiv0\pmod4$). 
	
	{\bfseries Subcase 2.3: $b=\omega^2$.}	 In the subcase, (\ref{BCT-uniform-inv-switch}) becomes 
\begin{equation}
	\label{BCT-uniform-inv-switch-2.3}
\left\{
\begin{array}{lr}
\frac{1}{x+\omega} + \frac{1}{y+\omega} = \omega, \\
\frac{1}{x}  + \frac{1}{y} = \omega. 
\end{array}
\right.
\end{equation}

	(i) If $x=0, 1, \omega, \omega^2$, $y=\omega^2,\omega, 1, 0$ respectively. Thus $(0,\omega^2), (\omega^2,0), (1, \omega) $ and $(\omega,1)$ are four solutions of (\ref{BCT-uniform-inv-switch-2.3}).
	
	(ii) If $x\neq 0, 1, \omega, \omega^2$, we have 
\begin{equation}
\label{BCT-uniform-inv-switch-2.3-1}
\left\{
\begin{array}{lr}
\frac{1}{x+\omega}+ \frac{1}{y+\omega} = \omega^2, \\
\frac{1}{x}  + \frac{1}{y} = \omega^2.
\end{array}
\right.
\end{equation}

	After simplifying (\ref{BCT-uniform-inv-switch-2.3-1}), we get $y=\frac{x}{\omega^2x+1}$ and $x^2+\omega x+\omega^2=0$. Then $x=1$ or $\omega^2$, which is a contradiction.

	{\bfseries Subcase 2.4: $b\neq 1, \omega, \omega^2$.} In the subcase, (\ref{BCT-uniform-inv-switch}) becomes 
\begin{equation}
\label{BCT-uniform-inv-switch-2.4}
\left\{
\begin{array}{lr}
	f(x+\omega) + f(y+\omega) = b, \\
f(x) + f(y) = b.
\end{array}
\right.
\end{equation}
	
	(i) If $x=0$, we get $f(y)=b+1$ and $f(y+\omega)=b+\omega^2$. Thus $y=\frac{1}{b+1}=\frac{1}{b+\omega^2}+\omega$, which means $b^2+\omega b+\omega = 0.$ Similarly, if $x=1$, we can get $y=\frac{1}{b}$ and $b^2+\omega b+1=0$. Also, if $x=\omega$, we obtain $y=\frac{1}{b+\omega^2}$ and $b^2+\omega b+\omega = 0$. In addition, if $x=\omega^2$, we have $y=\frac{1}{b+\omega}$ and $b^2+\omega b+1=0$. 
	
	Hence, when $b^2+\omega b+\omega = 0$,  $(0,\frac{1}{b+1})$, $(\frac{1}{b+1}, 0)$,  $(\omega, \frac{1}{b+\omega^2})$ and $(\frac{1}{b+\omega^2}, \omega)$ are four  solutions of (\ref{BCT-uniform-inv-switch-2.4}); when $b^2+\omega b+1=0$, $(1,\frac{1}{b})$, $(\frac{1}{b},1)$, $(\omega^2,\frac{1}{b+\omega}), (\frac{1}{b+\omega}, \omega^2)$ are four solutions of (\ref{BCT-uniform-inv-switch-2.4}).
	
	(ii) If $x\neq 0, 1, \omega, \omega^2$, (\ref{BCT-uniform-inv-switch-2.4}) becomes 
\begin{equation}
\label{BCT-uniform-inv-switch-2.4-1}
\left\{
\begin{array}{lr}
\frac{1}{x+\omega} + \frac{1}{y+\omega} = b, \\
\frac{1}{x}  + \frac{1}{y} = b.
\end{array}
\right.
\end{equation}

	After simplifying (\ref{BCT-uniform-inv-switch-2.4-1}), we get $y=\frac{x}{bx+1}$ and $bx^2+b\omega x+\omega=0$, which has at most two solutions in $\gf_{2^n}$. Moreover, (\ref{BCT-uniform-inv-switch-2.4-1}) has at most two solutions in $\gf_{2^n}\times\gf_{2^n}$.
	
	Therefore, in the case, $a=\omega$ or $\omega^2$,  the maximum of numbers of solutions in $\gf_{2^n}\times\gf_{2^n}$ of (\ref{BCT-uniform-inv-switch}) is
	\begin{equation*}
	\left\{
	\begin{aligned}
	4,  ~~\text{if}~~ n\equiv2\pmod4, \\
	6,   ~~\text{if}~~ n\equiv0\pmod4.
	\end{aligned}
	\right.
	\end{equation*}

	{\bfseries Case 3: $a\notin\{1,\omega,\omega^2\}$. }
	
	(i) If $x = 0, 1, a, a+1$,  according to Remark \ref{rem-8-4},  when $n\equiv0\pmod3$, (\ref{BCT-uniform-inv-switch}) has at most eight solutions in $\gf_{2^n}\times\gf_{2^n}$ and when $n\not\equiv 0\pmod 3$, (\ref{BCT-uniform-inv-switch}) has at most four solutions in $\gf_{2^n}\times\gf_{2^n}$. 
	
	(ii) If $x\neq 0, 1, a, a+1$, (\ref{BCT-uniform-inv-switch}) becomes 
\begin{equation}
\label{BCT-uniform-inv-switch-3}
\left\{
\begin{array}{lr}
\frac{1}{x+a} + \frac{1}{y+a} = b, \\
\frac{1}{x}  + \frac{1}{y} = b.
\end{array}
\right.
\end{equation}

	After simplifying (\ref{BCT-uniform-inv-switch-3}), we get $y=\frac{x}{bx+1}$ and $x^2+ax+\frac{a}{b}=0$, which has two solutions $x_0, x_0+a$ in $\gf_{2^n}$ if and only if $\tr_{2^n}(\frac{1}{ab})=0$ and when $\tr_{2^n}(\frac{1}{ab})=0$, $(x_0, x_0+a)$ and $(x_0+a,x_0)$ are two solutions of (\ref{BCT-uniform-inv-switch-3}). 
	
	From Lemma \ref{lem-S1S2S3S4}, (\ref{BCT-uniform-inv-switch}) has eight solutions in $\gf_{2^n}\times\gf_{2^n}$ if and only if $n\equiv0\pmod3$ and $(a,b)\in\{(a,b) \in\gf_{2^n}^{*}\times\gf_{2^n}^{*} | a^3+a+1=0 ~~\text{and}~~ b=a \}.$ From the above (ii),  $(x_0, x_0+a)$ and $(x_0+a,x_0)$ are two solutions of (\ref{BCT-uniform-inv-switch-3}) if and only if $\tr_{2^n}(\frac{1}{ab})=0$ holds. In the following, we consider when the above two conditions hold at the same time. Plugging $b=a, a^3+a+1=0$ into  $\tr_{2^n}(\frac{1}{ab})=0$, we get $$\tr_{2^n}\left(\frac{1}{a}\right)=\tr_{2^n}\left(1+a^2\right)=\tr_{2^n}(1+a)=0.$$
	On the other hand, since $a^3+a+1=0$, i.e., $a^4+a^2+a=0$, $\tr_{2^n}(a)=(a+a^2+a^4)+(a+a^2+a^4)^{2^3}+\cdots+(a+a^2+a^4)^{2^{n/3-1}}=0.$ Thus  $\{(a,b) \in\gf_{2^n}^{*}\times\gf_{2^n}^{*} | a^3+a+1=0 ~~\text{and}~~ b=a \}\cap \{(a,b) \in\gf_{2^n}^{*}\times\gf_{2^n}^{*} | \tr_{2^n}\left(\frac{1}{ab}\right)=0 \} \neq \emptyset$ if and only if $n\equiv0\pmod3$ and $\tr_{2^n}(1)=0$, i.e., $n\equiv0\pmod6$.  In other words, in the case, for any $n\equiv0\pmod6$, (\ref{BCT-uniform-inv-switch-3}) has at most ten solutions in $\gf_{2^n}\times\gf_{2^n}$ and for any $n\equiv3\pmod6$, (\ref{BCT-uniform-inv-switch-3}) has at most eight solutions in $\gf_{2^n}\times\gf_{2^n}$.  
	
	In addition, (\ref{BCT-uniform-inv-switch}) has four solutions in $\gf_{2^n}\times\gf_{2^n}$ if and only if $(a,b)\in (S_3\cup S_4) \backslash (S_3\cap S_4),$ where $S_3 := \{ (a,b)\in\gf_{2^n}^{*}\times\gf_{2^n}^{*} |  a^2b^2+a^2b+ab+1=0 \}$ and $S_4 := \{ (a,b) \in\gf_{2^n}^{*}\times\gf_{2^n}^{*} | a^2b^2+ab^2+ab+1=0  \}.$  In the following, we consider $S\cap \{(a,b)\in\gf_{2^n}^{*}\times\gf_{2^n}^{*} | \tr_{2^n}\left(\frac{1}{ab}\right) = 0 \}$, where $S = (S_3\cup S_4) \backslash (S_3\cap S_4)$ and $n\not\equiv0\pmod3$.   Since $\tr_{2^n}\left(\frac{1}{ab}\right)=0$ holds if and only if there exists $z\in\gf_{2^n}$ such that $\frac{1}{ab}=z+z^2$, i.e., $ab=\frac{1}{z}+\frac{1}{z+1}$. Plugging $ab=\frac{1}{z}+\frac{1}{z+1}$ into $a^2b^2+a^2b+ab+1=0$, we have $\frac{1}{z^2}+\frac{1}{z^2+1}+(\frac{1}{z}+\frac{1}{z+1})(1+b)+1=0,$ i.e., $$b=\frac{z^4+z+1}{z^2+z}.$$  Let $b=\frac{z^4+z+1}{z^2+z}$ and $a=\frac{1}{z^4+z+1}$. Obviously, for any $n\not\equiv0\pmod3$, there exist some $z\in\gf_{2^n}$ such that $a, b$ exist. Therefore,  $S\cap \{(a,b) \in\gf_{2^n}^{*}\times\gf_{2^n}^{*} | \tr_{2^n}\left(\frac{1}{ab}\right) = 0 \}\neq \emptyset$. That is to say, for any $n\not\equiv0\pmod3$, in the case, (\ref{BCT-uniform-inv-switch-3}) has at most six solutions in $\gf_{2^n}\times\gf_{2^n}$. 
	
	Therefore, in this case $a\notin\{1,\omega,\omega^2\}$,  the maximum of numbers of solutions in $\gf_{2^n}\times\gf_{2^n}$ of (\ref{BCT-uniform-inv-switch}) is
	\begin{equation*}
	\left\{
	\begin{aligned}
	10,  ~~\text{if}~~ n\equiv0\pmod6, \\
	8, ~~\text{if}~~ n\equiv 3 \pmod 6, \\
	6,   ~~\text{if}~~ n\not\equiv0\pmod3.
	\end{aligned}
	\right.
	\end{equation*}

	To sum up, together with Cases 1, 2 and 3, we know that the maximum of numbers of solutions in $\gf_{2^n}\times\gf_{2^n}$ of (\ref{BCT-uniform-inv-switch}) is
	\begin{equation*}
	\left\{
	\begin{aligned}
	10,  ~~\text{if}~~ n\equiv0\pmod6, \\
	8, ~~\text{if}~~ n\equiv 3 \pmod 6, \\
	6,   ~~\text{if}~~ n\not\equiv0\pmod3.
	\end{aligned}
	\right.
	\end{equation*}
\end{proof}

\end{document}